\documentclass[11pt]{article}
\usepackage[top=1in, bottom=1in, left=1in, right=1in]{geometry}
\usepackage{booktabs} 
\usepackage[ruled]{algorithm2e} 

\SetAlFnt{\small}
\SetAlCapFnt{\small}
\SetAlCapNameFnt{\small}
\SetAlCapHSkip{0pt}
\IncMargin{-\parindent}

\usepackage{amsmath}
\usepackage{amsthm}
\usepackage{amsfonts}
\usepackage{amssymb}
\usepackage{algpseudocode}
\usepackage{mathtools}
\usepackage{bm}
\usepackage{hyperref}
\usepackage{graphicx}
\usepackage{tikz}
\usetikzlibrary{positioning}
\hypersetup{
    colorlinks=true,
    linkcolor=blue,
    filecolor=magenta,      
    urlcolor=cyan,
    pdftitle={Overleaf Example},
    pdfpagemode=FullScreen,
    }

\newtheorem{lemma}{Lemma}
\newtheorem{prop}{Proposition}
\newtheorem{theorem}{Theorem}

\newtheorem{cor}{Corollary}
\newtheorem{definition}{Definition}

\renewcommand{\bar}{\overline}
\renewcommand{\hat}{\widehat}

\DeclareMathOperator*{\argmax}{arg\,max}

\newcommand{\reals}{\mathbb{R}}

\newcommand{\ubar}{\bar{u}}

\newcommand{\M}{\mathcal{M}}
\newcommand{\V}{\mathcal{V}}
\newcommand{\A}{\mathcal{A}}
\newcommand{\D}{\mathcal{D}}
\newcommand{\Ss}{\mathcal{S}}

\newcommand{\thetahat}{\hat{\theta}}
\newcommand{\rhohat}{\hat{\rho}}
\newcommand{\muhat}{\hat{\mu}}
\newcommand{\mhat}{\hat{m}}

\newcommand{\I}{\mathcal{I}}

\newcommand{\ur}{w}

\newcommand{\up}{u}

\newcommand{\vhat}{\hat{v}}
\newcommand{\VT}{V^\Theta}
\newcommand{\QM}{Q^\mathcal{M}}
\newcommand{\QV}{Q^\mathcal{V}}
\newcommand{\QT}{Q^\Theta}

\newcommand{\QVij}{Q^\mathcal{V}_{\vhat,v}}
\newcommand{\QTij}{Q^\Theta_{\thetahat,\theta}}

\newcommand{\squishlist}{
   \begin{list}{$\bullet$}
    { \setlength{\itemsep}{0pt}      \setlength{\parsep}{3pt}
      \setlength{\topsep}{3pt}       \setlength{\partopsep}{0pt}
      \setlength{\leftmargin}{1.5em} \setlength{\labelwidth}{1em}
      \setlength{\labelsep}{0.5em} } }
\newcommand{\squishend}{  \end{list}  }

\DeclareMathOperator{\E}{\mathbb{E}}

\algdef{SE}[SUBALG]{Indent}{EndIndent}{}{\algorithmicend\ }%
\algtext*{Indent}
\algtext*{EndIndent}

\usepackage[numbers]{natbib}
\setcitestyle{number}
\newcount\Comments  
\usepackage{color} 
\newcommand{\kibitz}[2]{\ifnum\Comments=1\textcolor{#1}{#2}\fi}

\begin{document}
\title{A Persuasive Approach to Combating Misinformation}
\author{        
        Safwan Hossain$^{*}$ \\ 
        Harvard University \\ 
        \texttt{shossain@g.harvard.edu}
        \and
        Andjela Mladenovic\thanks{Equal Contribution} \\ 
        UdeM and Mila\\ 
        \texttt{andjela.mladenovic@mila.quebec} 
        \and
        Yiling Chen \\ 
        Harvard University \\ 
        \texttt{yiling@seas.harvard.edu} 
        \and 
        Gauthier Gidel \\ 
        UdeM and Mila\\
        \texttt{gauthier.gidel@mila.quebec}
}

\date{}

\newcommand{\hugh}[1]{\textcolor{red}{\emph {\textbf{ Hugh:}} {#1}}}

\maketitle

\begin{abstract}
    Bayesian Persuasion is proposed as a tool for social media platforms to combat the spread of misinformation. Since platforms can use machine learning to predict the popularity and misinformation features of to-be-shared posts, and users are largely motivated to share popular content, platforms can strategically signal this informational advantage to change user beliefs and persuade them not to share misinformation. We characterize the optimal signaling scheme with imperfect predictions as a linear program and give sufficient and necessary conditions on the classifier to ensure optimal platform utility is non-decreasing and continuous. Next, this interaction is considered under a performative model, wherein platform intervention affects the user's future behaviour. The convergence and stability of optimal signaling under this performative process are fully characterized. Lastly, we experimentally validate that our approach significantly reduces misinformation in both the single round and performative setting and discuss the broader scope of using information design to combat misinformation. 
\end{abstract}

\maketitle

\section{Introduction}
Spreading misinformation has been one of the major critiques of social media platforms. The structure of these platforms rewards users for sharing content to gain popularity (e.g. number of likes and future re-shares), often irrespective of its veracity. This has prompted debates on how platforms should police their content. The most common approach is using fact-checking to tag and censor untruthful content. However, such approaches can hardly keep up with the speed and scale of today's content generation and the validity of content may not be a binary true or false. Censorship moreover leads to thorny debates around the platform regulating freedom of speech. While most Americans support taking steps to restrict misinformation, half of Americans in 2021 agreed that ``freedom of information should be prioritized over ... restricting false information online''~\citep{PewResearch}.  

We present information design as a viable approach, either as an alternative or complement, for addressing misinformation on platforms. Our approach leverages the information asymmetry between a platform and its users. The platform strategically reveals some information that users care about, who act in self-interest given the information. When properly designed, this strategic revelation can lead to harmful/poor content being shared less frequently. This approach acknowledges the different incentive structures for the user and platform and does not require the platform to police already shared posts; instead, the platform disincentivizes the sharing of misinformed content by the user.

Specifically, we model the platform-user interaction under the Bayesian Persuasion (BP) framework ~\cite{kamenica2011bayesian}. A post has a two-dimensional hidden state: its popularity if shared and its degree of misinformation. While a user has some prior belief over these states, she naturally does not know the true popularity of her to-be-shared content.
She may also be unaware of her post's level of misinformation as it is often introduced unintentionally and driven by social-psychological factors such as a sense of belonging~\cite{Wardle2020}, confirmation bias~\cite{Ecker2022}, and habitual sharing~\cite{Ceylan2023}. Moreover, the user may also be indifferent to the misinformation state and care only about the popularity their post achieves if shared. While the user may be indifferent or unaware of misinformation, the platform's utility for user action (sharing or not sharing) depends on the true realization of both states --- the post's popularity and degree of misinformation. This misalignment of the user's and the platform's utility poses a challenge for the platform: the platform does not want the user to share misinformed posts, but the user derives high utility by sharing popular content, which may not be truthful. The platform, however, possesses an informational advantage due to its vast troves of historical user and content data. So while the platform, like the user, does not know with certainty the true state of the post, it can leverage this to build classification models to predict (with certain error rates) the post's popularity and degree of misinformation. Based on the predictions, the platform can strategically reveal some stochastic information about the post's state, hoping to alter the user's belief about the post. The user, updating their belief based on the revealed information, can decide to share or not share the post to maximize her expected utility. The interaction of the platform and the user forms a Bayesian Stackelberg game, with the platform (leader) choosing an information revelation (signaling) scheme first, with the user (follower) deciding on an optimal action for herself based on the received information. The goal of the platform is to choose a signaling scheme to optimize the platform's expected utility at equilibrium and, by doing so, reduce misinformed content on the platform while maintaining engagement. The platform must be cognizant of the long-term dynamic of such interventions and its implicit effect on user behaviour over time. 


\paragraph{\textbf{Our contributions:}}
We propose an information design approach for social media platforms to address the spread of misinformed content. Platforms predict the properties of a user's to-be-shared post and strategically reveal this to users to improve the quality of shared content. This \emph{noisy persuasion} setting is formally defined in section \ref{sec:Model}. The lack of perfect information leads to an effective reduction of signaling power. Section \ref{section:prelim} discusses this alongside other preliminaries that generalize known results for the standard BP setting. The exact effect of the platform's classification accuracy on its optimal signaling scheme and the resulting utility is then discussed in section \ref{section:optimal_persuasion}. Specifically, we formulate this as a linear program and provide sufficient and necessary conditions on the classifier to ensure the optimal platform utility is monotone and continuous. In section \ref{section:performative}, the platform-user interaction is viewed through a performative lens, wherein signaling affects the user's sharing behaviour and correspondingly, their future beliefs about their content. We give a complete characterization, proving the stability and convergence of this performative process. Our findings are experimentally validated in section \ref{section:experiments}, with technical and conceptual extensions for tackling misinformation using information design discussed in section \ref{section:discussion}.


\section{Related Works}
This work proposes a soft approach, not involving censorship or tagging, for combating online misinformation, an approach advocated in the literature \citep{Howe2023, Jackson2022, Pennycook2022}. \citet{Howe2023} proposed to cap the depth (how many times messages can be relayed) or width (how many people a message can be shared with) of a social network to improve information fidelity. \citet{Jackson2022} observed that allowing people to only share posts that they have indicated are true, what they termed self-censoring, reduces the spread of misinformation. \citet{Pennycook2022} experimentally concluded that interventions that shifted users' attention toward the concept of accuracy could help reduce online misinformation sharing. These works, together with ours, all focus on reducing the sharing of misinformation. \citet{yang2023designing} and \citet{Candogan2020} respectively studied reducing the creation and the consumption of inaccurate information through signaling. Instead of considering an intervention, \citet{acemoglu2021model} modeled the propagation of misinformation as a game and analyzed its equilibrium, whereas \citet{Candogan2020} focuses on optimization with respect to externality effects. While spiritually similar, both assume platforms have perfect knowledge and do not consider performative effects. 

The seminal work on Bayesian persuasion \cite{kamenica2011bayesian} has led to many follow-up studies: the computational complexity of the sender's optimization problem \cite{dughmi2016algorithmic}, the impact of restricting signaling schemes and exogenous information distortions on the informativeness and welfare of equilibrium \citet{kosenko2021noisy, tsakas2021noisy}, and robust equilibrium concepts when sender has ambiguity over external information that receiver may have \cite{dworczak2022preparing}. Information design has also been used for studying price discrimination \cite{Bergemann2015}, multiplayer games \cite{Chenghan2022}, and revenue management \cite{Drakopoulos2021}. \citet{BergemannMorris2019}, \citet{Kamenica2019} and \citet{Candogan-bookchapter} offer comprehensive reviews on work in this area. Our work proposes information design as a tool for addressing misinformation on social media under a learned observation model. The insights within this model can also be considered spiritually similar to works on information ordering in economics literature \citep[Chapter 12]{bergin2005microeconomic}. 

Our dynamic setting in section \ref{section:performative} is inspired by the literature on performative prediction \cite{perdomo2020performative, mendler2020stochastic, mofakhami2023performative}. This models the phenomenon that when decision-making is influenced by predictions, the decision will affect future predictions. In our setting, it is not the predictions but rather the signaling process that leads to this dynamic. Our performative model here parallels the stateful one proposed by \citet{brown2022performative} and naturally captures the notion that both the past distribution and the user’s present sharing decisions (influenced by signaling) will affect content distribution on the platform. The platform should be cognizant of such long-term distribution changes when trying to influence user decisions. From a technical perspective, our results here are novel and not captured by existing results in the literature.  

\section{Model}\label{sec:Model}
Consider the interaction between a social media \emph{platform}, and a \emph{user}. Without any platform intervention, the user lacks additional information when they draft posts for submission; thus, they take their default action to share. Our model considers a platform predicting features of this draft and committing to revealing or \emph{signaling} this information to the user according to a randomized scheme. Signaling affects user's belief about their post's content, and thus their action to share or not. We now precisely define this, observing that this single-user model is without loss of generality since the platform could interact in such a way with multiple users.

\paragraph{\textbf{States and Predictions:}} Let $\M = \{1, \dots, m_{max}\}$ and $\V = \{1, \dots, v_{max}\}$ denote the possible misinformation and validation/popularity states respectively. 
A post drafted by a user has some true joint feature state $\theta = (m, v)$ drawn from a prior distribution $\mu \sim \Delta^{|\Theta|}$, where $\Theta = \M \times \V$ and $\Delta^k$ denotes the $k$ simplex. The prior encapsulates the distribution of user's shared posts. Both platform and users know this distribution since it is simply composed of past statistics of this user's content. However, the user does not observe the true state $\theta$ for their drafted content - i.e., they do not know the true popularity/misinformation of any drafted post a priori. The platform, however, can leverage their data and scale to predict these states of the draft content, with predictions denoted by $\thetahat = (\mhat, \vhat)$. Unlike the canonical BP settings, we do not assume these predictions to be perfect. We capture the inaccuracy of the prediction models used in our \emph{noisy persuasion} problem as follows:
\begin{definition}
    The platform's validation and misinformation classifier uncertainty is captured by the multi-class confusion matrices $\QV \in [0,1]^{|\V| \times |\V|}$ and $\QM \in [0,1]^{|\M| \times |\M|}$. Let $\QT = \QM \otimes \QV$ denote the combined confusion matrix, with $\otimes$ representing the Kronecker product. An element at index $(\thetahat, \theta)$ of $\QT$, denoted by $\QTij$, records $P(\thetahat| \theta)$.
\end{definition}
We assume the predictors $\mhat$ and $\vhat$ to be conditionally independent given $m, v$ and at least as good as chance - i.e. $P(\mhat=x|m=x) \geq \frac{1}{m_{max}}$ and same for $\QV$. 



\paragraph{\textbf{Actions and Utilities:}} Upon drafting content, the user can choose between one of two actions: $\A = \{0,1\}$. The action $0$ corresponds to the user \emph{not sharing} the content, and 1 corresponds to them \emph{sharing}. Both platform and user utility depend on this action and the underlying state of the content. We formally define them below:
\begin{definition}\label{definition:utility}
    We denote the user utility function as $\ur: \A \times \M \times \V \rightarrow \reals$ and the platform utility function as $\up: \A \times \M \times \V \rightarrow \reals$. We assume both utilities are bounded.
\end{definition}
Similar to the standard persuasion settings, we assume the platform to know the user's utility. This allows them to anticipate the user's best action and design a scheme accordingly. For users, a special utility structure is when they care only about the validation their posts receive and may be indifferent to or disagree with the platform's characterization of misinformation. Lastly, note that minimal restrictions are placed on platform utility and thus they are free to choose this to balance revenue/engagement and veracity as desired.


\paragraph{\textbf{Signaling Scheme:}} The platform maintains a set of \emph{signals} $\Ss$ and reveals a \emph{signaling scheme} before the user drafts any content. This is simply a set of conditional distributions, denoted by $\pi(s|\mhat, \vhat)$, which specifies the probability of sending signal $s$ when the platform predicts the draft content to have state $\thetahat = (\mhat, \vhat)$. Since scheme $\pi$ must be committed to a priori, the platform goal is to design $\pi$ to maximize their expected ex-ante utility, formally defined as:
\begin{definition}\label{definition:ex_ante_utility}
    The platform's \emph{ex-ante utility for a signaling scheme} $\pi(s|\thetahat)$ is $\sum_{s}{P(s)\E_{\rho^s}[\up(a^*, \theta)]}$, where $\rho^s(\theta) = P(\theta|s)$ is the posterior distribution induced by signal $s$ and scheme $\pi$, and $a^* = \argmax_{a}\E_{\rho^s}{[\ur(\theta, a)]}$ is the optimal receiver action for that posterior belief.
    \vspace{-0.5em}
\end{definition}
We now summarize the platform-user interaction for an \emph{instance} $\I = (\up, \ur, \mu)$ as follows:
\squishlist
    \item Platform reveals a signaling scheme $\pi(s|\thetahat)$ 
    \item User drafts content with unknown state $\theta = (m,v) \sim \mu$
    \item Platform uses learning models with joint confusion matrix $\QT$ to obtain prediction $\thetahat = (\mhat, \vhat)$ for this post, and then samples signal $s \sim \pi(s|\thetahat)$.
    \item User observes the signal, computes their posterior belief, and takes their optimal action $a^*$.
    \item User attains utility $\ur(a^*, \theta)$ and platform attains utility $\up(a^*, \theta)$.
\squishend
Game theoretically, this interaction outlines a Stakelberg game with the platform and user taking on the leader and follower roles respectively. The signaling scheme maximizing the platform's ex-ante utility given the user best-responds (definition \ref{definition:ex_ante_utility}) is the Stakelberg equilibrium strategy. A key thrust of our work is understanding how properties of the prediction accuracy (i.e. confusion matrix $\QT$) affect this optimal signaling scheme and the resulting platform utility. 

\paragraph{\textbf{Performative Model:}}
Without persuasion, the user takes their utility maximizing action based on their prior: $a^* = \argmax_{a}\E_{\mu_0}{\ur(a, \theta)}$, which we naturally assume to be ``share''. As such, the distribution of content shared over time matches the original prior. Applying persuasion affects this, however, and the user's decision to share now depends on the signaling scheme; thus, the distribution of a user's shared posts changes over time due to this intervention. We model this through a performative angle. At round $t$, $\mu_t$ represents the distribution of the user's currently published content. The platform deploys a signaling scheme $\pi_t(s|\thetahat)$, and the user observes recommendations from this whenever she drafts content. Note that draft posts that were persuaded to not be shared are absent from the platform, and older content loses relevance. As such, we model the content distribution for round $t+1$ as a convex combination of the present prior $\mu_t$ and the content that was shared this round, $P_t(\theta|a^*=1)$. This is formally presented in Section \ref{section:performative}.

\section{Preliminaries}\label{section:prelim}
We begin by observing that the noisy setting restricts the platform's signaling power compared to the standard setting. We then establish two characterizations, one on the signal space at optimal signaling and the other on optimal sender utility, which parallel the characterizations in standard Bayesian persuasion. We then give a didactic example that synthesizes these observations and shows that persuasion improves platform utility and reduces misinformation.


\subsection{Noisy Persuasion is Less Powerful}
Compared with standard BP where the sender observes the realization of $\theta$ perfectly and signals based on $\theta$, our setting is noisy as the platform only learns and signals based on $\thetahat$. In this noisy setting, the set of \emph{effective signaling} schemes - i.e. equivalent signaling schemes that are based on $\theta$ - the platform can use is a subset of those in standard noiseless BP settings. Indeed, for a given prior $\mu$ and confusion matrix $\QT$, a noisy signaling scheme $\pi(s|\thetahat)$ is equivalent to a signaling scheme $\Tilde{\pi}(s|\theta) = \sum_{\thetahat}\pi(s|\thetahat)\QT_{\thetahat, \theta}$ in standard BP since both induce the same posterior beliefs over $\theta$.

While any signaling scheme over noisy observations can be transformed into one over exact observations, the converse is not true. That is, there exists schemes $\tilde{\pi}(s|\theta)$ that cannot be expressed in terms of $\pi(s|\thetahat)$. To see this intuitively, when the platform observes the state $\theta$ perfectly, it can signal $\tilde{\pi}(s=\theta|\theta) =1$ to fully reveal the state to the user.
However, if the platform itself does not observe $\theta$ perfectly, it can never induce such a certain belief in the user. In other words, noisy persuasion (when $\QT \ne \I$, since $\I$ corresponds to perfect predictions/standard BP) has a smaller set of effective signaling schemes at its disposal. While the comparison to standard persuasion is intuitive, a more prescient question is how does the signaling space and crucially the platform's optimal expected utility change between two arbitrary confusion matrices $\QT_1$ and $\QT_2$? This is a more involved question which we answer in Theorem \ref{theorem:monotonicity} by giving a strict ordering for confusion matrices.

\subsection{Simplifying the Signaling Scheme}\label{sec:revelation}

The platform in general can use any set of signals. However, we next show in Proposition \ref{prop:revelation_principle} that only $|\A|$ signals are needed to attain its best possible ex-ante expected utility in noisy persuasion (hereinafter referred to as \emph{optimal platform utility}), just as in standard BP with perfect observations \cite{kamenica2011bayesian, dughmi2016algorithmic}. Since $|\Ss| = |\A|$ suffices, signals can, without loss of generality, be interpreted as recommending a specific action, providing significant operational simplicity (proofs of this section are in Appendix \ref{Appendix:A}).  

\begin{prop}\label{prop:revelation_principle}
    For instance $\I = (\up, \ur, \mu)$ and joint confusion matrix $\QT$, let $u_\I^*(\QT)$ represent the optimal platform utility achievable with an arbitrary number of signals. Then it is also possible for the platform to achieve $u_\I^*(\QT)$ utility using exactly $|\A|$ signals (i.e. $|\Ss| = |\A|$). 
\end{prop}


\subsection{Geometry of Noisy Persuasion}\label{section:geometry}
For any scheme, the posterior distributions induced by the signal realization $s$, denoted by $\rho^s$, must always satisfy $\sum_{s}{P(s)\rho^s} = \mu$, a condition termed \emph{Bayes Plausibility}. \citet{kamenica2011bayesian} show that optimal signaling can be interpreted as inducing a platform-advantageous set of Bayes plausible beliefs. Formally, they construct a mapping from belief to expected sender utility and show that the optimal sender utility is equivalent to evaluating the concave closure of this function at the prior. We now show that this observation can be generalized to our setting with predicted states, a valuable insight for forthcoming results.

\begin{definition}
    Let $\rhohat$ denote a belief over predicted states $\thetahat$, and $\rho$ a belief over true states $\theta$. Then for instance $\I$, define $\bar{\up}(\rhohat): \Delta^{|\Theta|} \rightarrow \reals$ as mapping from $\rhohat$ to the platform expected utility (w.r.t. to the corresponding true belief $\rho$) for the optimal user action at that belief: $\E_{\rho}[\up(a^*(\rho), \theta)]$. Let $co(\rhohat)$ denote the convex hull of the graph of $\bar{u}(\rhohat)$, and $cl(\rhohat) = \sup\{\bar{\up}(\rhohat) | (\rhohat, z) \in co(\rhohat)\}$ its concave closure.
\end{definition}

To interpret this, for a belief over predicted states $\rhohat$, the corresponding belief over true states $\rho$ can be computed through a linear map $\VT$ (see Lemma \ref{lemma:observed_to_true}). Then $a^*(\rhohat)$ denotes the optimal user action for the corresponding true belief $\rho$. With $a^*(\rhohat)$, one can compute the platform's expected utility $\bar{\up}(\rhohat)$ over the corresponding $\rho(\theta)$. $co(\rhohat)$ denotes the convex hull of all such $(\rhohat, \bar{\up}(\rhohat))$ points, with the concave closure $cl(\rhohat)$ representing the boundary of this convex hull. Lastly, Bayes Plausibility can be re-stated as $\sum_{s}{P(s)\rhohat(s)} = \muhat_0$, where $\muhat$ is the corresponding belief over predicted states for prior $\mu$. Proposition \ref{theorem:concave_closure} now relates the optimal platform utility to the concave closure, generalizing the result of \citet{kamenica2011bayesian} to the noisy setting.


\begin{prop}\label{theorem:concave_closure}
    The platform utility achieved by optimal signaling on instance $\I$ is equal to $cl(\muhat)$, where $\muhat = \sum_{\theta}{\mu(\theta)\QTij}$ and $\mu(\theta)$ is the prior.
\end{prop}

\begin{cor}\label{corollary:user_never_decrease}
    The expected user utility can never decrease due to any signaling scheme. 
\end{cor}

Now consider a function $w_a(\rhohat)$ denoting the user's expected utility for taking action $a$ at belief $\rhohat$. This is a linear function since the mapping from $\rhohat(\thetahat)$ to $\rho(\theta)$ is linear (Lemma \ref{lemma:noisy_to_true_scheme}), and expectation is linear. By analogously defining $\bar{\ur}(\rhohat) = \max(w_0(\rhohat), w_1(\rhohat))$ for the user (their expected utility at a belief due to optimal action), we note this is a convex function. Due to Bayes plausibility, the expected user utility after signaling is $\sum_{s}{P(s) \bar{\ur}(\rhohat_s)}$, which by the convexity of $\bar{\ur}$ can never be lower than the utility at the prior (no signaling), leading to Corollary \ref{corollary:user_never_decrease}. For the platform, while their utility never decreases, strict improvement is equivalent to $\ubar(\muhat) < cl(\muhat)$. Intuitively, if the platform utility $\ubar(\muhat)$ is aligned with the user along some direction in the belief space, then it can always signal to reveal more information along that direction and strictly improve. Platform and user alignment in the belief space are generally satisfied as both parties would prefer to share popular and true posts and not share false unpopular ones. Indeed, alignment is even easier if the user utility is also cognizant of reducing misinformation, and not purely popularity-driven. 

\subsection{Example}
\begin{figure}[ht]
    \centering
    \includegraphics[width=0.6\linewidth]{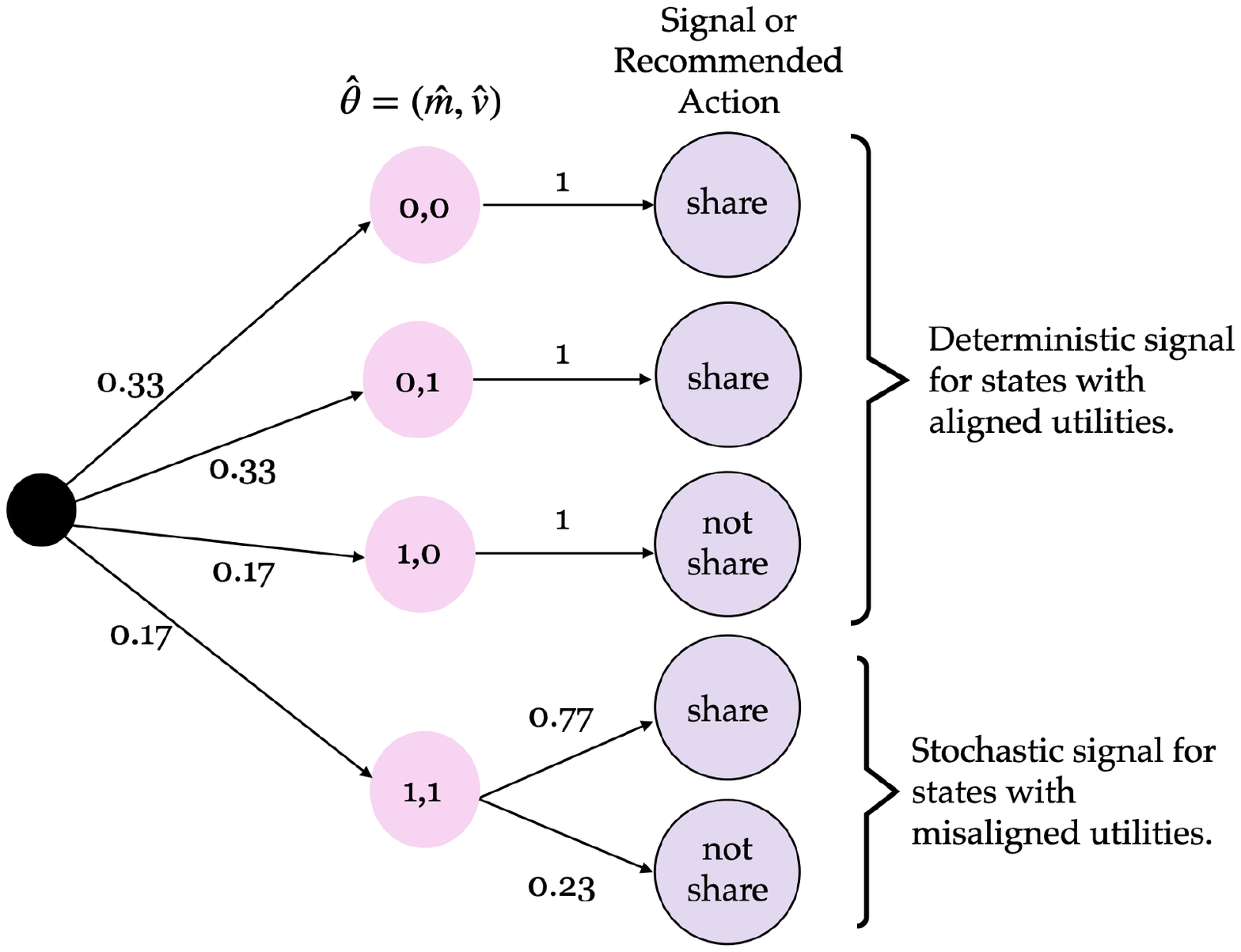}
    \caption{Noisy prior and signaling for example instance. Edges represent probability}
    \label{fig:example_instance}
\end{figure}

Consider an instance with binary validation and misinformation states ($0$ is not true/not popular and $1$ is false/popular), with $1$ and $0$ denoting user action ``share'' and ``not share''. We consider a popularity-driven user who is indifferent to misinformation, which as discussed is the harder setting. Let the prior $\mu$ and platform and user utilities $u$ and $w$ be:  

\begin{table}[h]
    \centering
    \vspace{-4pt}
    \begin{tabular}[\textwidth]{cccccc}
        \toprule
       \small$\theta = (m,v)$  &\small $\mu(\theta)$ &\small $u(\theta, 0)$ &\small $u(\theta, 1)$ &\small $w(\theta, 0)$ &\small $w(\theta, 1)$\\
       \midrule
        \small$(0,0)$ &\small 0.35 &\small 0 &\small 1 & \small0 &\small -1\\
        \midrule
        \small$(0,1)$ &\small 0.35 & \small-1 &\small 2 &\small -2 &\small 1\\
        \midrule
        \small$(1,0)$ &\small 0.15 &\small 0 &\small -1 & \small0 &\small -1\\
        \midrule
        \small$(1,1)$ & \small0.15 &\small 0 & \small-3 & \small-2 &\small 1\\ 
        \bottomrule
    \end{tabular}
    \label{tab:my_label}
\end{table}

Without signaling, the optimal user action at the prior is to share, whose outcomes are given in the first row of Table \ref{tab:example_instance}. Suppose the platform has a $90\%$ accurate classifier for both $m$ and $v$. The prior over-predicted states $\thetahat$ are presented in Fig. \ref{fig:example_instance} (first set of edges). Consider the signaling scheme specified here, with the platform fully revealing information in aligned states and strategically obscuring when misaligned. When the user behaves optimally under this signaling, their expected utility is unchanged, but the platform's expected utility and the fraction of shared content that is misinformation drastically improves (Table \ref{tab:example_instance}).

\begin{table}[htb]
    \centering
    \begin{tabular}{lccc}
    \toprule
    {} &     \parbox{1.5cm}{\raggedright {Platform utility}} &  \parbox{3cm}{{\% of shared post that is misinfo.}} & \parbox{0.7cm}{\raggedright{User \\ utility}} \\
    \midrule
    \parbox{3.5cm}{{Before persuasion}}      &  {0.45} &  {30} &{0}\\  
    \midrule
     \parbox{3.5cm}{{After persuasion}}         &  {0.64}&  {17} & {0} \\
    \bottomrule
    \end{tabular}
    \caption{\small {Results before and after persuasion}}
    \label{tab:example_instance}
\end{table}

\section{Optimal Noisy Persuasion}\label{section:optimal_persuasion}
While the geometric perspective gives us insight into the properties of signaling, it does not offer a way to compute an optimal signaling scheme. In this section, we begin by providing a linear program (LP) that computes the optimal signaling scheme under inaccurate predictions. We then characterize how the resulting optimal platform utility is affected by the confusion matrix $\QT$. Specifically, we give an ordering of matrices $\QT$ as it relates to effective signaling space, provide necessary and sufficient conditions on $\QT$ for optimal platform utility to be non-decreasing, and show that this optimal utility is Lipschitz continuous in $\QT$. These results are not only structural but also of operational significance since classifier accuracy is something platforms can modify and improve, making it important to understand the dynamic. Proofs for this section are in Appendix \ref{Appendix:B}.

As shown in Section~\ref{sec:revelation}, it suffices to focus on signaling schemes where $\Ss = \mathcal{A}$, with each signal interpreted as an action recommendation. When the platform commits to a signaling scheme $\pi(s|\thetahat)$, its effective signaling scheme is $\tilde{\pi}(s|\theta)=\sum_{\thetahat}\pi(s|\thetahat)\QT_{\thetahat, \theta}$. 
Similar to \citet{dughmi2016algorithmic} who formulated an LP for optimal signaling scheme in standard BP, we formulate the following LP for solving the optimal signaling scheme in our noisy persuasion setting:      
 \begin{eqnarray}\label{LP:1}
  \textrm{max} & \sum_{a_i}^{|\A|}\sum_{\theta}{\up(a_i, \theta)\mu(\theta)\tilde{\pi}(s=a_i|\theta)}  &\\\label{LP:2}
  \textrm{s.t.} & \sum_{\theta}{\Delta \ur_{ij}(\theta)\mu(\theta)\tilde{\pi}(s=a_i|\theta)} \geq 0  & \forall a_i, a_j\\\label{LP:3}
                      & \Tilde{\pi}(s=a_i|\theta) = \sum_{\thetahat}\pi(s=a_i|\thetahat)\QT_{\thetahat, \theta} & \forall \, a_i, \theta\\\label{LP:4}
                      & \sum_{a_i}{\pi(s=a_i|\thetahat)} = 1& \forall \thetahat\\\label{LP:5}
                      & \pi(s=a_i|\thetahat) \geq 0 & \forall \, a_i, \thetahat 
  \end{eqnarray}
  where $\Delta \ur_{ij}(\theta) = \ur(a_i, \theta) - \ur(a_j, \theta)$.
(\ref{LP:2}) is the incentive compatibility constraint, which enforces that the recommended action has a higher expected utility under the induced posterior than any other action, making it optimal for the user. This phenomenon is often referred to as \emph{optimal signaling is persuasive}. Objective (\ref{LP:1}) then captures the platform's ex-ante expected utility when $\tilde{\pi}$ is the induced effective signaling scheme. (\ref{LP:1}) and (\ref{LP:2}) is the same as the LP for standard Bayesian persuasion. However, our LP for noisy persuasion requires additional constraints (\ref{LP:3}), (\ref{LP:4}), (\ref{LP:5}), which restrict $\tilde{\pi}$ to the set of effective signaling schemes that can be induced under confusion matrix $\QT$.

We now wish to understand how for a given instance $\I$, the quality of the platform classifier affects optimal achievable utility, $u^*_\I(\QT)$. While it is natural that a better classifier would lead to higher utility, there is no unique notion of ``better'' for a multi-class classifier. For example, entropy, recall, and F1 score are all different notions of classifier quality. We precisely settle this question for symmetric $\QT$ in Theorem~\ref{theorem:monotonicity} (Proofs for this section are in Appendix \ref{Appendix:B}) by leveraging our LP to prove that the optimal platform utility is non-decreasing exactly with respect to the set inclusion of the convex hull of the rows $[\QT]_{1:},\ldots,[\QT]_{|\Theta|:}$ of $\QT$. Importantly, our proof shows that if the rows of $\QT_1$ are contained within the convex hull of rows of $\QT_2$, \emph{any} effective signaling for $\QT_1$ can also be achieved under $\QT_2$. If this is not satisfied, then we show that an instance can always be constructed such that the \emph{optimal signaling} under $\QT_1$ corresponds to an effective signaling not possible under $\QT_2$, and vice versa. By taking the contrapositive, it also gives an ordering of confusion matrices as it relates to the effective signaling space, formally presented in Corollary~\ref{cor:ordering}.



\begin{theorem}\label{theorem:monotonicity}
    Given two symmetric confusion matrices $\QT_1$ and $\QT_2$ and any instance $\I$, the optimal sender utility, $\up^*_\I(\QT_2) \geq \up^*_\I(\QT_1)$, is non-decreasing if and only if $co([\QT_1]_{1:},\dots ,[\QT_1]_{|\Theta|:}) \subseteq co([\QT_2]_{1:},\dots ,[\QT_2]_{|\Theta|:})$, where $co([\QT_i]_{1:},\dots ,[\QT_i]_{|\Theta|:})$ is the convex hull of the rows of matrix $\QT_i$.
    
\end{theorem}

\begin{cor}\label{cor:ordering}
    Let $\Phi_{\I}(\QT)$ denote the set of all effective signaling schemes $\tilde{\pi}$ that noisy persuasion with confusion matrix $\QT$ can achieve. Then for two symmetric confusion matrices $\QT_1, \QT_2$, we have $\Phi_{\I}(\QT_1) \subseteq \Phi_{\I}(\QT_2)$ if and only if $co([\QT_1]_{1:},\dots ,[\QT_1]_{|\Theta|:}) \subseteq co([\QT_2]_{1:},\dots ,[\QT_2]_{|\Theta|:})$.
\end{cor}

Theorem \ref{theorem:monotonicity} is also operationally insightful since the conditions for optimal platform utility to monotonically increase can be easily parsed from the $\QT$ matrix, noting that convex-hull inclusion is easily verifiable. It also gives platforms a clear metric for comparing classifier quality for the noisy persuasion task. We next present the second key result of this section which shows that the change in optimal utility due to confusion matrix $\QT$ is Lipschitz continuous. This is advantageous operationally since making slight modifications to the underlying classifier will not drastically affect the optimal persuasion utility. The proof relies on the geometric insights developed in section \ref{section:prelim}. To sketch, we show that the concave closure function $cl(\rhohat)$ is Lipschitz in $\rhohat$, and that the change from $\QT_1$ to $\QT_2$ leads to a bounded vertical and lateral shift in the closure function due to this property, giving rise to Lipschitz continuity in $\QT$.

\begin{theorem}\label{theorem:continuity}
    For any $\I$, the maximum platform utility, $\QT \mapsto \up^*_{\I}(\QT)$, is Lipschitz continuous.
\end{theorem}


\section{Performative Perspectives}\label{section:performative}
We now consider the dynamics of noisy persuasion over time. Recall that without persuasion, users simply take the best action for their prior, which we naturally assume to be "share". As such, the distribution of content in the platform is unchanged from the starting prior. With persuasive signaling, we now stochastically induce different beliefs in the user, and their decision to share is affected accordingly. Over time, the content distribution this user forms their belief upon will skew toward the type of content they were recommended to share. Formally, the prior distribution between two-time steps, $\mu_t$ and $\mu_{t+1}$, is affected by the signaling employed by the platform.\footnote{The actual time between $t$ and $t+1$ is unimportant for our technical analysis and we leave it for practitioners at deployment.} Correspondingly, the platform's optimal signaling at $t+1$ will change from time $t$. In this section, we consider the utility function and predictions models fixed and look to model this interaction between optimal signaling and the resulting prior it leads to.

The nascent literature around performative prediction captures a similar tension in the classification setting \citep{perdomo2020performative, mendler2020stochastic, mofakhami2023performative}. Therein, an optimal classifier is deployed for a given distribution, which is then affected by the chosen classifier. While sharing several parallels, these works generally require the underlying optimization problem to be unconstrained and strongly convex, which while reasonable for classification, do not hold for our optimal signaling linear program. Further, they model the distribution update to be purely based on the optimization variable, where the past has no effect. Transition in our social media setting is not necessarily stateless since older posts still exist on the platform, albeit with diminished relevance. Inspired by \citet{brown2022performative}, we consider a stateful performative model where we interpolate between the earlier prior and the new distribution affected by signaling. We describe this below:
\begin{definition}\label{definition:performative_process}
    The performative persuasion process for an instance $\I = (u, w, \mu_0)$ with joint confusion matrix $\QT$ is defined as follows:
    \squishlist
        \item At each round $t$, with prior $\mu_t(\theta)$, the platform chooses an optimal signaling scheme $\pi_t^*(s|\thetahat)$, $s \in \{0, 1\}$.
        \item Since optimal signaling is persuasive, users follow the recommendation (to share or not to share).
        \item The next round's prior distribution is given by: $\mu_{t+1} = \lambda \mu_t + (1-\lambda)\rho_t(\theta|s=1)$, where $\rho_t(\theta|s=1)$ is the distribution of content that was shared this round, and $\lambda \in [0,1]$ a hyper-parameter.
    \squishend
\end{definition}

The key questions in this setting are if and how such a process converges. Parallel to this, another important notion in performative literature is \emph{stability}. Interpreted in our setting, this is a prior belief wherein optimal signaling does not change the underlying distribution. We precisely define both for our setting below: 

\begin{definition}
    The performative process converge to distribution $\mu^*$ if for any $\varepsilon > 0$, there exists a $T_c$ such that for $t > T_c$, $\mu_{t} - \mu^* < \varepsilon$. A distribution $\mu^s$ is denoted stable if $\mu_t = \mu^s$ implies $\mu_{t+1} = \mu^s$. 
\end{definition}

We first prove that in a stateful setting, $\lambda > 0$, the performative process always converges to the optimal posterior belief induced by the share signal (referred to as the share posterior) at round $0$ with prior $\mu_0$. This proof relies on the geometric insight developed in section \ref{section:geometry}. We leverage the fact that the performative process does not affect the concave closure function, and show that Bayes plausibility implies $\mu_t$ always lies between the first two optimal posteriors induced at the first round. All proofs for this section are in Appendix \ref{appendix:C}.

\begin{theorem}\label{theorem:performative_convergence}
    For $\lambda > 0$, the performative process always converges to the best optimal posterior induced by the share signal in the first round: $\rho^1_0(\theta) = \sum_{\thetahat}{\rhohat^1_0 \VT_{\theta, \thetahat}}$, which has utility $\bar{\up}(\rhohat^1_0)$.\footnote{$\rhohat^1_0$ is the posterior induced by signal $1$ (share) at round 0. which has utility $\bar{\up}(\rhohat^1_0)$}
\end{theorem}


The convergence of this process has some nice properties. First, the platform can easily determine the optimal utility at convergence since it is simply the utility of the share posterior induced by optimal signaling in the first round. Thus, given an instance $\I$ and their classifiers, the platform can easily determine how beneficial long-term persuasion will be for a given user. Our proof in fact shows that the optimal platform utility at each prior in this performative process is monotonically increasing. This is a by-product of the geometry of the $\ubar(\rhohat)$ function and the fact that $\mu_t$ always lies between the optimal posteriors induced at round $0$. However, this result implies nothing about stability, since $\mu_t$ is always changing, albeit slightly for large $t$. This leads us to answer in Theorem \ref{theorem:performative_stability} the remaining technical questions: (1) what is the convergence behavior when $\lambda = 0$ and (2) does this process admit a stable point? This proof also relies on the geometric insights of optimal noisy persuasion. 

\begin{theorem}\label{theorem:performative_stability}
    If $\lambda = 0$, the performative process converges in $2(|\Theta| - 1)$ iterations to a stable point where the platform always recommends to share. 
\end{theorem}

We note that the stable point above is indeed stable for any $\lambda$. More broadly, while having exactly $\lambda = 0$ is unlikely, this result can be interpreted as follows: the convergence process in Theorem \ref{theorem:performative_convergence} gets infinitely close to some distribution (rate depends on $\lambda$) but never reaches it. Since, realistically, there is a finite resolution and practical agents have bounded rationality, it may be reasonable to assume $\mu_t$ to reach $\rho^1_0$. Such a phenomenon is inextricably linked to the $\lambda = 0$ case for which the process directly moves to this $\rho^1_0$. Thus, the two results can be seen from a hierarchical perspective: Theorem \ref{theorem:performative_stability} can be seen as charting the high-level behavior of the performative process until reaching a stable point while Theorem \ref{theorem:performative_convergence} provides evidence that such a phenomenon also occurs in the more realistic case where performativity is a smoother process where $\mu_t$ depends on the previous priors. Alternatively, we can see the $\lambda = 0$ process as changing the user's priors with persuasion until reaching a point wherein optimal signaling can no longer achieve any benefit. This leads to our last technical question: what are the conditions wherein the stable point has higher platform utility than the original prior $\mu_t$. We provide a sufficient theoretical condition, we experimentally observe the stable point to be largely beneficial for the platform in practice (Fig. \ref{fig:performative_persuasion}).

\begin{theorem}\label{theorem:performative_monotonic}
    For $c_n = \max_{\theta}{\up(0, \theta)}$, define a normalized platform utility $\up'(a, \theta) = \up(a, \theta) - c_n$.\footnote{Recall the 0 action is to \emph{not share}} Then the performative process with $\lambda = 0$ always converges to a stable point in a monotonically increasing fashion if the sender's normalized utility at $\mu_0$ is positive. 
\end{theorem}

\section{Experiments}\label{section:experiments}
We now experimentally validate our approach: specifically, while we provide detailed theoretical results on the platform utility under optimal signaling, it is instructive to see how this translates into reducing misinformation sharing. Due to a lack of public data, we create a synthetic dataset for the three components of a noisy persuasion instance: the prior distribution, platform utility, and user utility. Three possible states are considered for validation and misinformation respectively, with 0 representing unpopular/true, 2 representing popular/false, and 1 representing a  middle ground. We sample utilities uniformly under the following natural constraints: the platform has 0 utility for the ``not share'' action, a positive utility that is increasing in $v$ for $m=0$, a negative utility that is decreasing in $v$ for $m=2$, and $\forall v \,,\, \up(a=1, m=0, v) > \up(a=1, m=1, v) > \up(a=1, m=2, v)$. Note that for $m=1$, the platform utility could either increase or decrease with $v$, mirroring the possible trade-offs platforms could make between revenue and content quality when uncertain. The user is assumed to be purely popularity-driven with 0 utility for not sharing and increasing utility in $v$ (from negative for $v=0$ to positive for $v=2$) for sharing; as discussed in section \ref{section:geometry}, this is usually the harder setting for persuasion. For each instance, we vary the classification error between 0 to 0.4 (the error is equally divided amongst all the incorrect classes) and plot in Fig. \ref{fig:single_round_persuasion} the decrease in the percentage of shared posts that are misinformation ($m=2$) before and after persuasion.

\begin{figure*}[htb]
\flushleft
\begin{minipage}{0.44\linewidth}
\centering
\begin{tikzpicture}
  \node (img)  {\includegraphics[width=1.0\linewidth]{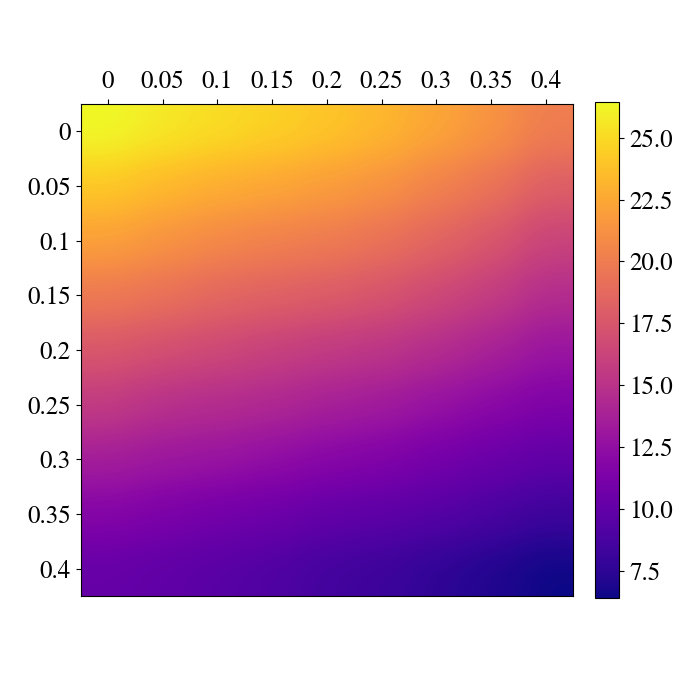}};
  \node[above=of img, node distance=0cm, yshift=-1.7cm,font=\color{black}] {V classifier error};
  \node[left=of img, node distance=0cm, rotate=90, anchor=center,xshift=0cm, yshift=-1cm,font=\color{black}] {M classifier error};
 \end{tikzpicture}
 \vspace{-40pt}
  \caption{Avg \% decrease in misinformation shared due to single application of noisy persuasion.}\label{fig:single_round_persuasion}
\end{minipage}%
\hspace{1.7cm}
\begin{minipage}{0.44\linewidth}
\begin{tikzpicture}
  \node (img)  {\includegraphics[width=1.0\linewidth]{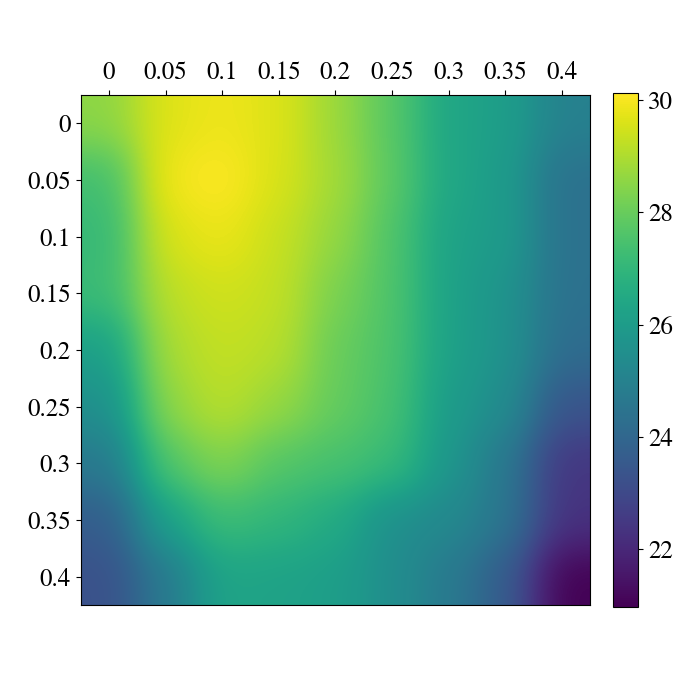}};
  \node[above=of img, node distance=0cm, yshift=-1.7cm,font=\color{black}] {V classifier error};
  \node[left=of img, node distance=0cm, rotate=90, anchor=center,yshift=-1cm,font=\color{black}] {M classifier error};
\end{tikzpicture}
\vspace{-40pt}
\hspace{20pt}\caption{Avg \% decrease in misinformation shared between prior and stable point}\label{fig:performative_persuasion}
\end{minipage}%
\vspace{-1em}
\end{figure*}


Even when both classifiers perform poorly, a $10\%$ misinformation reduction is achieved, increasing to around $20\%$ when the classifier errors are below $0.15$. We also note that the results are more sensitive to the accuracy of the misinformation classifier than the validation one. This is expected since the platform utility ordering essentially flips based on the $m$ state, making it crucial for platforms to capture the true misinformation level more accurately. 

Regarding the performative dynamic, note that Theorem~\ref{theorem:performative_convergence} implies that the prior at convergence for $\lambda \ne 0$ is the induced posterior corresponding to the user sharing $(a=1)$ at round $t=0$. Thus, the misinformation shared at the performative convergent belief is also captured by Fig.~\ref{fig:single_round_persuasion}. When $\lambda = 0$, Theorem~\ref{theorem:performative_stability} shows the process converges to a stable point and Fig.~\ref{fig:performative_persuasion} plots the average decrease in misinformation between the starting distribution and the stable distribution. While the results are less correlated with classifier accuracy, there is a substantial decrease in misinformation shared at the stable point, suggesting that the stable point is an improvement over the prior, and that long-term application of noisy persuasion is beneficial for the platform. Lastly, both plots have a $90\%$ confidence interval of less than $4\%$.  

\section{Discussion}\label{section:discussion}
This work takes a softer approach toward addressing misinformation on social media platforms by leveraging an information design framework based on Bayesian persuasion. Our setting, wherein underlying states are not perfectly observed but predicted, generalizes this classical framework to a noisy and realistic setting. We rigorously characterize how the prediction accuracy affects optimal signaling and platform utility, providing operationally useful results to platforms while noting that user utility can never decrease due to this. Further, the techniques used provide significant insights on noisy persuasion from a geometric and optimization perspective which may be of broader interest. We also consider the long-term implications of such an intervention and model the platform-user interactions from a performative angle. We rigorously characterize the convergence and stability properties of this process; these results illustrate that persuasion can have a long-term positive impact on content quality and veracity on social media platforms.

Our work leaves open a number of technical and conceptual questions. Providing necessary conditions that ensure increased platform utility at a stable point would complement Theorem \ref{theorem:performative_monotonic} and be an insightful result. As would augmenting our experiments with real user data and interactions to evaluate the real-world effectiveness of such an approach. Along a similar line, our model assumes a perfectly rational user with the sender not having any exogenous restrictions; generalizing this to consider a bounded rationality model \citep{jones1999bounded, de2022non} or designing robust signaling under exogenous restrictions \citep{dworczak2022preparing, kosenko2021noisy} would be an intriguing research direction. The robust persuasion model of \citet{dworczak2022preparing} may be especially relevant for the misinformation setting since it considers an exogenous third party also influencing the user and proposes optimal signaling for the worst case. 
Another interesting direction is how persuasion impacts influence propagation and network effects within platform \citep{barbieri2013topic, arieli2022herd}. Lastly, developing broader socio-technical guidelines around information design for online interactions is a prescient and necessary direction which we leave for future work.

\bibliography{bibliography}
\clearpage

\appendix
\onecolumn

\section{Appendix A}\label{Appendix:A}


\begin{lemma}\label{lemma:noisy_to_true_scheme}
    For a signaling scheme $\pi(s|\thetahat)$, the probability of observing a signal $s$ given true realization $\theta$ is given by $P(s|\theta) = \sum_{\thetahat}{\pi(s|\thetahat)\QTij}$. Further, $P(s|m,v) = \sum_{\mhat, \vhat}{\pi(s|\mhat, \vhat)\QM_{\mhat, m}\QV_{\vhat, v}}$
\end{lemma}
\begin{proof}
    The following is due to total probability law: $P(s|\theta) = \sum_{\thetahat}{P(s|\theta, \thetahat)P(\thetahat|\theta)}$. Note that given $\thetahat$, signal $s$ is conditionally independent of $\theta$ since signaling is directly specified by the former. We also note that the classification predictors $\mhat, \vhat$ are assumed to be conditionally independent given $m, v$, with $\QT$ being a Kronecker product. Thus, $P(s|m,v) = \sum_{\mhat, \vhat}{\pi(s|\mhat, \vhat)P(\mhat|m)P(\vhat|v)}$.
\end{proof}

\begin{lemma}\label{lemma:observed_to_true}
    Given belief over true states $\rho(\theta)$, the corresponding belief over predicted states is $\rhohat(\thetahat) = \sum_{\theta}{\rho(\theta)\QTij}$. Similarly, for belief $\rhohat(\thetahat)$, the corresponding belief over true states is: $\rho(\theta) = \sum_{\thetahat}{\rhohat(\thetahat)\VT_{\theta, \thetahat}}$, where the $\VT_{\theta, \thetahat} = \frac{\QTij \mu(\theta)}{\sum_{\theta'}{\QT_{\thetahat, \theta'}\mu(\theta')}}$.
\end{lemma}
\begin{proof}
     First, it is easy to see that $\rhohat(\thetahat) = \sum_{\theta}{P(\thetahat|\theta)\rho(\theta)} = \sum_{\theta}{\rho(\theta)\QTij}$. 
     For the other direction, note that $\rho(\theta) = \sum_{\thetahat}{\rhohat(\thetahat)P(\theta|\thetahat)}$. Next, by Bayes rule, it is evident that $P(\theta|\thetahat) = \frac{\QTij \mu(\theta)}{\sum_{\theta'}{\QT_{\thetahat, \theta'}\mu(\theta')}} \triangleq \VT_{\theta, \thetahat}$.
\end{proof}

\subsection*{Proof of Proposition \ref{prop:revelation_principle}}
\begin{proof}
    Let $\pi^*$ denote the optimal unrestricted signaling scheme with a total of $\ell$ signals. 
    We can state the posterior over true states $\theta = (m,v)$ for signal realization $s$ as: $P(\theta|s) = \frac{1}{P(s)} \mu(\theta)\sum_{\thetahat}{\pi^*(s|\thetahat)\QTij}$. For each $a$, let $S_a$ denote the set of signals whose induced posterior under $\pi^*$ leads to optimal action $a$. Next, consider a signaling scheme that directly recommends an action, satisfying $|\Ss| = |\A|$. Define this as follows: $\pi'(a|\thetahat) = \sum_{s \in S_a}{\pi^*(s|\thetahat)}$. Next, observe that utility at this optimal scheme, denoted by $u_\I^*(\QT) = \sum_{a}\sum_{s \in S_a}{P(s)}\sum_{\theta}{\up(a, \theta)P(\theta|s)} = \sum_{a, \theta}\sum_{s \in S_a}{\up(a, \theta)\mu(\theta)P(s|\theta)}$. \, We next use lemma \ref{lemma:noisy_to_true_scheme} and write this as equal to:
    \begin{gather*}
        \sum_{a, \theta, \thetahat}{\mu(\theta)\up(a, \theta)\QTij}\sum_{s \in S_a}{\pi^*(s|\thetahat)} = \sum_{a, \theta}\mu(\theta)\up(a, \theta)\sum_{\thetahat}{\QTij \pi'(a|\thetahat)}
    \end{gather*}
    We note that the inner summand (over $\thetahat$) is equivalent to $P(s=a|\theta)$ under the new action-recommending signaling scheme. Thus we can write the expected utility under the new signaling is: $\sum_{a}P(s=a)\sum_{\theta}{\up(a, m, v)P(\theta|a)} = u_\I^*(\QT)$ completing the proof.
\end{proof}

\subsection*{Proof of Proposition \ref{theorem:concave_closure}}
\begin{proof}
    Given a prior $\mu(\theta)$ over true states $\theta$, $\muhat$, which we call the noisy prior, can be interpreted as the corresponding belief over predicted states. The Bayes plausibility condition, which immediately follows from Bayes rule, implies that for any signaling scheme $\pi$, $\muhat(\thetahat) = \sum_{s}{P(s)\rhohat(\thetahat|s)}$, where $\rhohat(\thetahat|s)$ (also denoted by $\rhohat^s$) is the induced belief over predicted states upon receiving signal $s$. Thus, the expected utility of any signaling scheme must be in the set $\{\bar{\up}(\muhat) | (\hat{\mu}, \bar{\up}(\muhat)) \in co(\hat{P})\}$, since this includes all convex combinations of induced beliefs that equal the noisy prior, and their corresponding expected platform utility (which is simply the same convex combination of expected utilities at those beliefs). Thus $z^* = \sup\{\bar{\up}(\muhat) | (\hat{\mu}, \bar{\up}(\muhat)) \in co(\hat{P})\} = cl(\muhat)$ is the maximum utility achievable by any signaling scheme with an arbitrary number of signals. By proposition \ref{prop:revelation_principle} we know there exists a signaling scheme with $|S| = |A|$ that can also achieve this utility.
\end{proof}

Using the notation from section \ref{section:prelim}, recall $w_a(\rhohat)$ denotes the user's expected utility for taking action $a$ at belief $\rhohat$. Then for a belief $\rhohat$, if $w_0(\rhohat) = w_1(\rhohat)$, then $\rhohat$ represents the threshold where the receiver's optimal action changes. This threshold is a hyperplane since it is the intersection of two hyperplanes, and we call this the \emph{indifference plane} $\D$, since the user is indifferent to both actions at this point. \footnote{As standard in persuasion literature, we assume when the user is tied, it is broken in favour of the platform}. Since the $\bar{\up}(\rhohat)$ function is based on the optimal user action at belief $\rhohat$, $\bar{\up}(\rhohat)$ is possibly discontinuous over this indifference plane since this is where the optimal action for the user changes. We now show in Lemma \ref{lemma:indifference} that posteriors induced by a strictly optimal signaling scheme (strictly improves upon the platform utility at the prior belief) are inextricably linked to the indifference plane.

\begin{lemma}\label{lemma:indifference}
    The concave closure $cl(\rhohat)$ is continuous, piecewise linear, and possibly non-differentiable over the indifference plane. Further, the posteriors $\rhohat^s$ induced by strictly optimal signaling is either on the simplex boundary or the indifference plane, and $cl(\rhohat^s) = \bar{\up}(\rhohat^s)$.
\end{lemma}
\begin{proof}
    We observe that $\bar{\up}(\rhohat)$ is piece-wise linear due to the mapping from $\rhohat(\thetahat)$ to $\rho(\theta)$ being linear and expectation being a linear operator, with possible discontinuities at beliefs wherein the receiver is indifferent. Since $co(\rhohat)$ is the convex hull of this piecewise linear function defined for all $\rhohat$, and $cl(\hat{\mu})$ corresponds to the boundary of this convex hull, this must be continuous, piecewise linear and possibly non-differentiable over the indifference planes.
    
    Next, invoking lemma 4 in \citet{kamenica2011bayesian} implies that at any induced posterior $\rhohat^s$ for an optimal scheme, either (1) the belief is on the boundary, (2) the receiver must be indifferent to multiple actions at this belief, or (3) for any other belief wherein the receiver optimal action is not $a^*(\rhohat^s)$, it is strictly better for the platform that $a^*(\rhohat^s)$ is taken. Consider the posterior induced by signal 1, $\rhohat^1$, and suppose the user optimal action is $a$. Condition (3) implies that for all beliefs wherein the optimal receiver action is \emph{not share}, the platform would strictly prefer the share action. However, recall that in section \ref{sec:Model}, we assumed that for each action $a$, there is a state (and thus a corresponding belief) where the user and platform both prefer this action. Thus, only the first two can hold.

    Lastly, to show $cl(\rhohat^s) = \bar{\up}(\rhohat^s)$, consider the line segment $\ell_1$ connecting the induced posteriors of the optimal signaling scheme: $(\rhohat^0, \bar{\up}(\rhohat^0))$ and $(\rhohat^1, \bar{\up}(\rhohat^1))$. Observe that the optimal sender utility is obtained by evaluating this line segment at the noisy prior $\muhat$. If the line connecting these two points (including the end-points) is part of $cl(\rhohat)$, then our claim holds. If not, then this line must be in the interior of the convex hull since $cl(\rhohat)$ represents the boundary of this convex hull. Then, there exist points $(\rhohat'_0, \bar{\up}(\rhohat'_0))$ and $(\rhohat'_1, \bar{\up}(\rhohat'_1))$ such that the line between these two is strictly above $\ell_1$. Evaluating this line segment at the prior $\muhat$ (and thus satisfying Bayes plausibility) yields a strictly higher value than at $\ell_1$, contradicting our original claim that we start with an optimal signaling scheme. 

\end{proof}

As discussed in section \ref{section:prelim}, the harder scenario is when the user is purely popularity-driven and indifferent to misinformation since there is less alignment with the platform. Such users may indeed not even agree with the platform's characterization of misinformation. In these cases, we show below that while the platform can still signal conditioned on both $\thetahat = (\mhat, \vhat)$, it suffices to reveal the signaling over $\vhat$ to the user without loss of generality. This allows a nice operational simplicity.

\begin{prop}\label{prop:user_depends_validation}
    For a user whose utility only depends on the validation state, it suffices for the platform to reveal their marginal signaling scheme $\pi(s|\vhat)$ and $\QV$ to the user, for them to compute their true posterior $P(v|s) = \tfrac{1}{P(s)} \mu(v)\sum_{\vhat}{\pi(s|\vhat)\QV_{\vhat, v}}$.
\end{prop}
\begin{proof}
    Consider the platform revealing the marginal scheme $\pi(s|\vhat) = \sum_{\mhat}{\pi(s|\mhat, \vhat)}$. Since the signal only depends on $\vhat$, users can compute $P(s|v) = \sum_{\vhat}{P(s|v,\vhat)\QVij} = \sum_{\vhat}{\pi(s|\vhat)\QV_{\vhat, v}}$, and then compute the desired posterior over $v$  using the prior as follows: $P(v|s) = \tfrac{1}{P(s)}P(s|v)\mu(v)$.
\end{proof}

\section{Appendix B}\label{Appendix:B}

\subsection*{Proof of Theorem \ref{theorem:monotonicity}}

\begin{proof} 
    We first note that in our setting with binary actions, this optimal persuasion LP can be tersely expressed as follows, with $c$ and $B$ only depending on $\I$ (and not on $\QT$), and $\bm{\pi}$ and $\bm{\tilde{\pi}}$ denoting row vectors capturing the probability of sending signal 1 conditioned on noisy and true observations. 
     \begin{gather*}
         \text{maximize:} \langle c, \boldsymbol{\tilde \pi} \rangle \qquad \qquad \qquad \qquad \qquad \qquad \qquad \quad  \\
         \text{subject to: \,} \boldsymbol{\tilde \pi} B  \geq \bm{0} \text{\, and \,} \boldsymbol{\pi}\QT = \boldsymbol{\Tilde{\pi}} \text{\, and \,} \boldsymbol{0}\leq \boldsymbol{\pi} \leq \boldsymbol{1}
     \end{gather*}
     The set $\{\boldsymbol{\pi} \QT  \,\, | \,\, \boldsymbol{0}\leq \boldsymbol{\pi} \leq \boldsymbol{1}\}$ can be interpreted geometrically. This set corresponds to the \emph{parallelepiped} induced by the rows of $\QT$. In other words, for a set of vectors $\{\bm{v}_1, \dots,\bm{v}_k\}$, we define $paral(\bm{v}_1, \dots, \bm{v}_k)$ as the set of vectors that can be expressed as: $\beta_1 \bm{v}_1 + \dots + \beta_k \bm{v}_k$, with $\beta_i \in [0,1]$. One interpretation of this reformulated LP is that Bayesian persuasion with noisy observations can be seen as standard Bayesian persuasion with an additional constraint that the signaling scheme over the true observations (also referred to as effective signaling) $\boldsymbol{\tilde{\pi}}$ has to belong to the parallelepiped of the rows of the confusion matrix $\QT$. We now prove the stated theorem:

    $\Rightarrow$ We start by considering the sufficient condition for the optimal platform utility to increase wherein the rows of $\QT_1$ can be written as a convex combination of the rows of $\QT_2$. Denote the convex weights by row vectors
    $\boldsymbol{\lambda_1},\dots,\boldsymbol{\lambda_{|\Theta|}}$. Then the following holds: $L \QT_2 = \QT_1$, where $L$ is defined as $L = \begin{bmatrix}
        \boldsymbol{\lambda_1}\\
         \vdots \\
         \boldsymbol{\lambda_{|\Theta|}}\\
    \end{bmatrix}$.
    Notice that $L$ is a doubly-stochastic matrix. Specifically, observe that $\boldsymbol{1} L \QT_{2} = \boldsymbol{1} \QT_{1}=\boldsymbol{1}$, which implies $\boldsymbol{1} L = \boldsymbol{1}({\QT_2}^{-1})= \boldsymbol{1}$.

    Let $(\boldsymbol{\Tilde{\pi_1}}, \boldsymbol{\pi_1})$ denote any feasible solution for LP$(\QT_1)$. Specifically, $\boldsymbol{\Tilde{\pi_1}}$ is  a row vector whose coordinates are $\pi(1|\theta)$ and $\boldsymbol{\pi_1}$ is a row vector whose coordinates are $\Tilde{\pi}(1|\thetahat)$.  We will show that  $(\boldsymbol{\Tilde{\pi_2}}, \boldsymbol{\pi_2}) = (\boldsymbol{\Tilde{\pi_1}}, \boldsymbol{\pi_1}L)$ is a feasible solution for LP($\QT_2$). Notice that all constraints and objective of LP($\QT_2$) depend only on $\boldsymbol{\Tilde{\pi}}$, except following constraints: $\boldsymbol{\Tilde{\pi_2}} = \boldsymbol{\pi_2} \QT_2  \,\, \text{and} \,\, \boldsymbol{0}\leq \boldsymbol{\pi_2} \leq \boldsymbol{1}$. Note the first constraint is satisfied by $\boldsymbol{\pi_2} \QT_2  = \boldsymbol{\pi_1}L \QT_2  = \boldsymbol{\pi_1}\QT_1  = \boldsymbol{\Tilde{\pi_1}} = \boldsymbol{\Tilde{\pi_2}}$, while second constraint is 
\begin{equation}
    \boldsymbol{\pi_2} = \boldsymbol{\pi_1}L =\begin{pmatrix}
\langle  \boldsymbol{\pi_1}, [L]_{1:}  \rangle \\
\vdots\\
\langle
\boldsymbol{\pi_1}, [L]_{|\Theta|:}  \rangle
\end{pmatrix} 
\end{equation}
satisfied by the fact that $L$ is double stochastic matrix. 
Now, we can conclude that for optimal solution 
$(\boldsymbol{\Tilde{\pi_1}^{*}}, \boldsymbol{\pi_1}^{*})$ for LP$(\QT_1)$, we have that solution 
$(\boldsymbol{\Tilde{\pi_2}^{*}}, \boldsymbol{\pi_2}^{*}) = (\boldsymbol{\Tilde{\pi_1}^{*}}, \boldsymbol{\pi_1}^{*}L)$ is feasible solution for LP$(\QT_2)$ and that optimal platform utility for LP$(\QT_1)$ is always achievable when solving LP$(\QT_2)$. 
\newline
\newline
$\Leftarrow$
We now show the stated condition is necessary. That is, if the condition is not satisfied, there always exist instances wherein the utility is not monotone.  For symmetric confusion matrices, we first show that the rows of $[\QT_1]_{i:}$ belong to 
\begin{equation}
paral([\QT_2]_{1:},\dots ,[\QT_2]_{|\Theta|:}):= \{
     \beta_1[\QT_2]_{1:} + \dots + \beta_{|\Theta|}[\QT_2]_{|\Theta|:} \;:\; \beta_k \in [0,1] \}
\end{equation} also belong to $[\QT_1]_{i:} \in  co([\QT_2]_{1:}),\dots ,[\QT_2]_{|\Theta|:})$. Observe that $\sum_{j}{[\QT_1]_{i:j}} = 1$ and further, $\sum_{j}\sum_{\ell}{\beta_{\ell}{[\QT_2]_{\ell:j}}} = \sum_{\ell}\beta_{\ell}\sum_{j}{{[\QT_2]_{\ell:j}}} = \sum_{\ell}{\beta_\ell} = 1$, which is exactly the convex hull condition. Thus, the convex hull structure is equivalent to the parallelepiped structure when $\QT_1$ and $\QT_2$ are symmetric stochastic matrices.

Now, assume there exists a column of matrix $\QT_1$, such that $[\QT_1]_{i:} \notin$  $paral([\QT_2]_{1:},\dots ,[\QT_2]_{|\Theta|:})$ then we aim to show that there exists an instance $\I = (\up, \ur, \mu)$ such that $\up^*_\I(\QT_1) > \up^*_\I(\QT_2)$, violating the monotone condition. Consider an instance wherein the receiver is indifferent to both actions at all states. Thus, the first two constraints of LP$(\QT)$ are always satisfied.

Then for $\QT_1$, let $\phi_1$ denote the set of any $\bm{\tilde{\pi}}$ such that $\bm{\pi} \QT_1 = \bm{\tilde{\pi}}$, with $\bm{0} < \bm{\pi} < \bm{1}$, define $\phi_2$ similarly for $\QT_2$. Since there exists a column of matrix $\QT_1$, such that $[\QT_1]_{i:} \notin$  $paral([\QT_2]_{1:},\dots ,[\QT_2]_{|\Theta|:})$ we know that there exists a point, denoted by $\boldsymbol{\Tilde{\Tilde{\pi}}_1}$ which belongs to $paral([\QT_1]_{1:},\dots ,[\QT_1]_{|\Theta|:})$, but not $paral ([\QT_2]_{1:},\dots ,[\QT_2]_{|\Theta|:})$. Since $\boldsymbol{\Tilde{\Tilde{\pi}}_1} \in paral ([\QT_1]_{1:},\dots ,[\QT_1]_{|\Theta|:})$ we know that $\exists \boldsymbol{\pi_1}$ whose values are between $\bm{0}$ and $\bm{1}$ such that $\boldsymbol{\pi_1} \QT_1 = \boldsymbol{\Tilde{\Tilde{\pi}}_1}$. On contrary, since $\boldsymbol{\Tilde{\Tilde{\pi}}_1} \notin paral([\QT_2]_{1:},\dots ,[\QT_2]_{|\Theta|:})$ we know that $\nexists \boldsymbol{\pi_2}$ whose values are between 0 and 1 such that $\boldsymbol{\pi_2} \QT_2 = \boldsymbol{\Tilde{\Tilde{\pi}}_1}$. Therefore, we know that $\exists\boldsymbol{\Tilde{\Tilde{\pi}}_1}$ such that $\boldsymbol{\Tilde{\Tilde{\pi}}_1} \in \phi_1$ and $\boldsymbol{\Tilde{\Tilde{\pi}}_1} \notin \phi_2$. Now, by Hyperplane Separation Theorem~\citep[Exercise 2.22]{boyd2004convex}, we know that $\exists \boldsymbol{b}$ such that $\langle \boldsymbol{b} , \Tilde{\boldsymbol{\Tilde{\pi_1}}}\rangle > c_1$ and $\langle \boldsymbol{b} , \boldsymbol{\Tilde{\pi_2}}\rangle < c_2$, $\forall \boldsymbol{\Tilde{\pi}_2} \in \Phi_2$ such that $c_1 > c_2$.

Now notice that we can simply achieve our objective expression by $(\up(a_1,\theta)-\up(a_2,\theta))\mu(\theta) = b_\theta$ where $\theta = 1 \dots |\Theta|$, and $b_\theta$ is the $\theta$-th coordinate of vector $\boldsymbol{b}$. Notice that in the edge case, where $\mu(\theta) = 0$ for some $\theta$ values the initial LP problem reduces to the same LP structure with a smaller dimension, and the proof is exactly the same.
Therefore, we see that we can find examples where $\up^*_\I(\QT_1) > \up^*_\I(\QT_2)$.
\end{proof}

\subsection*{Proof of Theorem \ref{theorem:continuity}}
\begin{proof}
    We wish to show that for any instance $\I$ and any pair $(\QT_1, \QT_2)$, there exists a constant $L$ such that: $|\up^*_\I(\QT_1) - \up^*_\I(\QT_2)| \leq L||\QT_1 - \QT_2||_{\infty}$, where we use the $\ell_{\infty}$ matrix norm. We first show that the concave closure function $cl(\rhohat)$ is Lipschitz in $\rhohat$. Next, we show that the change from $\QT_1$ to $\QT_2$ leads to a bounded vertical shift in the closure function due to this Lipschitz property. Lastly, the point at which we evaluate the concave closure function to determine the optimal sender utility also changes in a bounded manner. The combined effects are all bounded and give rise to our Lipschitz constant. 

    From lemma \ref{lemma:indifference}, we know that on either side of the plane of indifference, $cl(\rhohat)$ is a continuous linear function. Since we have 2 actions, there is a single plane of indifference. We will aim to tightly upper bound the directional derivative for each joint state $\thetahat=(\mhat,\vhat)$ of the linear regions on either side of this indifference plane. Since the concave closure function is linear on either side of the indifference plane, 
    it suffices to consider the maximum value of $\frac{\Delta u}{\Delta \rhohat_{\thetahat}}$ from the indifference plane, where $\rhohat_{\thetahat}$ denotes the coordinate corresponding to $\thetahat$. Pick an arbitrary belief $\rhohat$ on the indifference plane. The slope along a direction $\thetahat$ is naturally upper bounded by $\frac{u_s^{max} - u_s^{min}}{\min(\rhohat_{\thetahat}, 1-\rhohat_{\thetahat})}$. We now look to tighten this. Indeed, if there exists another belief $\rhohat'$ on the indifference plane such that $\rhohat'_{\thetahat}$ is closer to 0.5, this would imply the change in utility is larger than $\up^{max} - \up^{min}$, which is not possible. To tighten this, let $\rhohat_{\thetahat}^{min}$ and $\rhohat_{\thetahat}^{max}$ denote the smallest and largest values of coordinate $\thetahat$ for beliefs $\rhohat$ that lie on the indifference plane. Then the largest value of the directional derivative of the closure function along the $\thetahat$ direction is given by $c_{\thetahat} = \max \left( \frac{\up^{max} - \up^{min}}{1 - \rhohat_{\thetahat}^{min}}, \frac{\up^{max} - \up^{min}}{\rhohat_{\thetahat}^{max}} \right)$. Thus, when belief changes from $\rhohat_1$ to $\rhohat_2$, the maximum change in the concave closure function is upper-bounded by $|\rhohat_1 - \rhohat_2|\sum_{\thetahat}{c_{\thetahat}}$, with $c = \sum_{\thetahat}{c_{\thetahat}}$ being the Lipschitz constant.

    For a belief over predicted states $P(\thetahat) = \rhohat$, the corresponding belief over true states $\rho = P(\theta)$ changes under the different confusion matrices. This affects the closure graph in two ways. First, it vertically shifts the underlying $\bar{\up}$ function since $\bar{\up}(\rhohat) = \E_{\rho(\theta)}[\up(a^*, m, v)]$. More formally, for a given belief $\rhohat$ the change due to the shift from $\QT_1$ to $\QT_2$ is given by: $\bar{\up}(\rhohat; \QT_1) - \bar{\up}(\rhohat; \QT_2)$ 
    \begin{equation*}
        = \sum_{\theta}\up(a^*, \theta)\sum_{\thetahat}{\rhohat(\thetahat)(\VT_1(\theta, \thetahat) - \VT_2(\thetahat, \theta))} \leq |\Theta|u^{max}||\VT_1 - \VT_2||
    \end{equation*}
    which follows due to lemma \ref{lemma:noisy_to_true_scheme}. Next, exploiting the relationship between matrices $\VT$ and $\QT$ as outlined in Lemma \ref{lemma:observed_to_true}, the following holds for any element $(\theta, \thetahat)$ of matrix $\VT$:
    \begin{equation}
        \VT_{\theta, \thetahat} = \frac{\mu(\theta)\QTij}{\sum_{\theta'}{\mu(\theta')\QT_{\thetahat, \theta'}}} \implies \left|\frac{\partial \VT_{\theta, \thetahat}}{\partial \QT_{\thetahat, \theta_i}}\right| \leq \frac{\mu(\theta_i)\mu(\theta)\QT_{\thetahat, \theta}}{\left(\sum_{\theta'}{\mu(\theta')\QT_{\thetahat, \theta'}}\right)^2} \leq \frac{|\Theta|^2}{\mu^2_{min}}
    \end{equation}
    where the last inequality follows since we consider classifiers to at at least as good as chance, the diagonals of $\QT$ are always at least $\tfrac{1}{|\Theta|}$. We also note that $\mu_{min} > 0$ since if any state occurs with 0 probability, we can without loss of generality, reformulate the problem to exclude that state. With the dependence of an element of $\VT$ established to be Lipschitz in an element of $\QT$, it follows that $||\VT_1 - \VT_2|| \leq M ||\QT_1 - \QT_2||$ where $M = \tfrac{|\Theta|^3}{\mu^2_{min}}$. 
    
    
    
    Second, let $\rho(\theta)$ denote a belief wherein the user has equal expected utility for both actions. In other words, the user is indifferent. The predicted state belief that corresponds to this true belief also shifts due to the change from $\QT_1$ to $\QT_2$. In other words, the indifference plane is laterally shifted. To precisely characterize this, observe that for a belief over true state $\rho(\theta)$, the change in corresponding predicted belief is: $\rhohat(\thetahat; \QT_1) - \rhohat(\thetahat; \QT_2) \leq |\Theta| ||\QT_1 - \QT_2||$ which also follows form lemma \ref{lemma:noisy_to_true_scheme}. There is thus a region of length $|\Theta| ||\QT_1 - \QT_2||$ where the optimal action is different under $cl(\rhohat; \QT_1)$ and, $cl(\rhohat; \QT_2)$. The largest difference between the two closure functions for a belief $\rhohat$ occurs in this region since one function could be increasing and the other decreasing (for a belief outside of this region, the optimal action is the same thus, both functions are both increasing or decreasing). However, since the closure function is lipschitz with constant $c$, the maximum difference at any belief $\rhohat$ is given by: $|cl(\rhohat;\QT_1) - cl(\rhohat;\QT_1)| \leq 2c|\Theta| ||\QT_1 - \QT_2|| + |\Theta|M u^{max}||\QT_1 - \QT_2||$. 
  
    Lastly, recall that by theorem \ref{theorem:concave_closure}, the optimal platform utility is equal to evaluating the closure function at the noisy prior $\muhat(\thetahat)$. For a prior $\mu(\theta)$, we have established that the predicted belief shifts at most $|\Theta|||\QT_1 - \QT_2||_1$. Since the closure function is Lipschitz, the change in the evaluation point from $\muhat(\thetahat\; \QT_1)$ to $\muhat(\thetahat\; \QT_2)$ leads to a difference of at most $c|\Theta| ||\QT_1 - \QT_2||_1$. Combining this with the fact that the closure functions are vertically by at most $(2c + M u^{max})|\Theta|||\QT_1 - \QT_2||_1$ at any belief, implies that the total change in optimal utility due to change from $\QT_1$ to $\QT_2$ is at most $|\Theta|(3c + M u^{max})||\QT_1 - \QT_2||_1$. Thus, the optimal sender utility for an instance $\I$ as a function of confusion matrix $\QT$ is $|\Theta|(3c + M u^{max})$-Lipschtiz. 
\end{proof}

\section{Appendix C}\label{appendix:C}
\subsection*{Proof of Theorem \ref{theorem:performative_convergence}}
\begin{proof}
    We first express the performative dynamics with respect to the observed states $\thetahat = (\mhat, \vhat)$. Observe the following: $\lambda \muhat_t(\thetahat) + (1-\lambda)\rhohat_t(\thetahat | s=1) = \sum_{\theta}{\QTij [\lambda \mu_t(\theta) + (1-\lambda)\rho_t(\theta|s=1)]}$, which follows due to lemma \ref{lemma:observed_to_true}. Next, by invoking the performative dynamics and lemma \ref{lemma:observed_to_true} again, we have that this is equivalent to $\sum_{\theta}{\mu_{t+1}(\theta)\QTij} = \muhat_{t+1}(\thetahat)$. In other words: $\muhat_{t+1}(\thetahat) = \lambda \muhat_t(\thetahat) + (1-\lambda)\rhohat_t(\thetahat|s=1)$. \footnote{}
    
    Next, observe that the performative process only changes the prior (and thus the optimal signaling scheme), but does not affect the platform belief to expected utility function $\bar{\up}(\rhohat)$ nor its concave closure $cl(\rhohat)$. Both of these depend only on the platform and user utility. Next, consider the scenario at $t=0$ with noisy prior $\muhat_0(\thetahat)$, whereupon the platform commits to an optimal signaling scheme $\pi^*_0(s|\thetahat)$. If there are multiple optimal schemes (thus multiple optimal posteriors that can be induced) and since this is the first round, let the platform break the tie by choosing the pair for whom the corresponding $\bar{\up}(\rhohat_0(\thetahat|s=1))$ is the largest\footnote{If multiple optimal signaling schemes exist at subsequent rounds, we assume ties are broken by choosing the scheme whose posteriors are closest to the earlier round's posterior.}. Denote this pair of optimal signaling schemes by $\{\rhohat^0_0, \rhohat^1_0\}$. Next, for $\alpha \in [0,1]$, consider the line segment $\ell(\alpha) = (1 - \alpha)z_0 + \alpha z_1$ which connects the points $z_0 = [\rhohat^0_0, cl(\rhohat^0_0)]$ and $z_1 = [\rhohat^1_0, cl(\rhohat^1_0)]$. By lemma \ref{lemma:indifference}, we know that the endpoints of this line segment, corresponding to the optimal induced posteriors, also lie on the $\bar{\up}(\rhohat)$. In other words, $[\rhohat^x_0, cl(\rhohat^x_0)] = [\rhohat^x_0, \bar{\up}(\rhohat^x_0)]$ for $x \in \{0,1\}$. By Bayes plausibility, there exists an $\alpha_0$ such that the first element of $\ell(\alpha_0)$, denoted by $\ell(\alpha_0)_0 = \muhat_0$. Theorem \ref{theorem:concave_closure} further states that the expected platform utility under the optimal scheme is given by $cl(\muhat_0) = \sum_{x \in {0,1}}{P_0(x)\bar{\up}(\rhohat^x_0)}$, where $P_0(x)$ is the probability of signal $x$ at round 0. Observe that this point is on the line segment $\ell$ -- i.e. $\ell(\alpha_0) = [\muhat_0, cl(\muhat_0)]$. Thus, we have that the line segment $\ell$ matches the $cl(\rhohat)$ at both the endpoints and one interior point. Since the $cl(\rhohat)$ function is concave piece-wise linear continuous, it must imply $cl(\rhohat)$ coincides with the line segment $\ell$ between beliefs $\rhohat^0_0$ and $\rhohat^1_0$. 

    The performative dynamic implies the next prior, $\muhat_1$ is a convex combination of $\muhat_0$ and $\rhohat^1_0$. Thus, there exists an $\alpha_1$ such that $\ell(\alpha_1)_0 = \muhat_1$, since the performative process moves to a point between the earlier prior and the $s=1$ posterior, both of which are on the $\ell$ line segment. More specifically, $\alpha_1 = (1 - \lambda) + \lambda \alpha_0$. The value of optimal signaling at $\muhat_1$ is again equivalent to $cl(\muhat_1)$ by theorem \ref{theorem:concave_closure}. Since the curve $\ell(\cdot)$ coincides with $cl(\cdot)$ between $\rhohat^0_0$ and $\rhohat^1_0$, and $\muhat_1$ lies in this region, the value of optimal signaling is equal to $\ell(\alpha_1)_1 = cl(\muhat_1)$. Thus, $\rhohat^0_0$ and $\rhohat^1_0$ are still optimal induced posteriors for prior at $\muhat_1$, and recall we break ties by choosing posteriors closest to the last round. In other words, $\rhohat^x_1 = \rhohat^x_0$, with $P_1(s=0) = 1 - \alpha_1$ and $P_1(s=1) = \alpha_1$ is an optimal scheme for the platform at $\muhat_1$. It is evident that this invariant will be maintained throughout this process. Precisely, for each $\mu_t$ in this process, there exists a $\alpha_t \in [0,1]$ such that $\ell(\alpha_t)_0 = \muhat_t$, which implies $cl(\mu_t) = \ell(\alpha_t)_1$, which implies signaling such that $\rhohat_t^x = \rhohat^x_0$, and $P_t(s=0) = 1 - \alpha_t$ and $P_1(s=1) = \alpha_t$ is optimal for the platform at $\mu_t$. This in turn means $\mu_{t+1}$ can be expressed as $\ell(\alpha_{t+1})_0$, with $\alpha_{t+1} = (1-\lambda) + \lambda \alpha_t$. Since $\alpha_t, \lambda \in [0,1]$, the following always holds: $\alpha_t < \alpha_{t+1} < 1$; thus, as $t \rightarrow \infty$, $\alpha_t \rightarrow 1$ and $\muhat_t \rightarrow \rhohat^1_0$.
\end{proof}

\subsection*{Proof of Theorem \ref{theorem:performative_stability}}
\begin{proof}
    Observe that any belief $\rhohat \in \Delta^{|\Theta|}$ can be equivalently represented as a vector in ${|\Theta|-1}$ dimensional space. As such, hereinafter, we will refer to $\rhohat \in \reals^{|\Theta|-1}$. Recall that $w_a(\rhohat)$ is a linear function representing the expected user utility at belief $\rhohat$ for taking action $a$, and the plane of indifference, which we denote by $\D$ is the intersection of these two planes, with $dim(\D) \leq |\Theta|-1$. If $w_0(\rhohat)$ and $w_1(\rhohat)$ are co-planer (thus like $\rhohat$, $dim(\D) = |\Theta|-1$), the user is indifferent at \emph{all} beliefs $\rhohat$. Since ties break in favour of the platform, $\bar{\up}(\rhohat) = max(\up_0(\rhohat), \up_1(\rhohat))$, where $\up_a(\rhohat)$ denotes the expected platform utility when action $a$ is taken at belief $\rhohat$. Since the mapping from $P(\thetahat)$ to $P(\theta)$ is continuous (lemma \ref{lemma:noisy_to_true_scheme}) and expectation is linear, $\up(\rhohat, x)$ is linear and $\bar{\up}(\rhohat)$ is continuous and piecewise linear. Further, since this is a max over two linear functions, $\bar{\up}(\rhohat)$ is either convex or concave everywhere. As observed by \citet{kamenica2011bayesian}, when this function is concave, signaling cannot strictly improve utility, and when convex, the optimal induced posterior is at the simplex boundary. Note that when signaling cannot improve utility, uninformative signaling that does not change the prior is optimal, and we reach a stable point. We now prove our claim inductively. 
     
    We start with the base case of $|\Theta| = 2$ with $\rhohat \in [0,1]$. If the user is indifferent everywhere, the above implies that either signaling cannot improve utility and we are at a stable point or signaling moves the share posterior to the simplex boundary, addressed below. Otherwise, $\D$ is simply a point on this interval. Starting at any noisy prior $\muhat_t$, the first optimal signaling step will result in induced posteriors either at point $\D$, or on the boundary $[0,1]$ due to lemma \ref{lemma:indifference}. Note that since $\lambda = 0$, the prior at the next round is the ``share'' induced posterior, $\rhohat^1_0$. Thus, the next round's prior will be either on the boundary (i.e. $0$ or $1$) or at $\D$. If the prior is at the boundary, observe that there exists no signaling scheme that can strictly improve upon the utility at the prior, since there is no convex combination of distinct posteriors that can equal the prior. Thus, we have stability at round $t+1$. If the next round prior is at $\D$, optimal signaling at round $t+1$ must (due to lemma \ref{lemma:indifference}) either stay at $\D$ (and thus achieve convergence), or move the boundary and achieve stability in the next round. Thus, when $|\Theta| = 2$, a stable point is achieved in at most 2 steps.

    Our inductive hypothesis is the following: For a state size $|\Theta| = n \geq 2$, the performative process converges in at most $2(|\Theta|-1)$ in the $\lambda = 0$ regime. Now consider an instance with $|\Theta| = n + 1$, and accordingly, $dim(\D) \leq n$. Suppose we are at time $t$ with prior $\muhat_t$. If $dim(\D) = n$, the user is indifferent everywhere, and we know that either signaling cannot improve utility and we have reached a stable point, or signaling will move the share posterior (and thus $\mu_{t+1}$) to the simplex boundary, a case addressed below. Otherwise, $dim(\D) \leq n-1$, and after the first optimal signaling, lemma \ref{lemma:indifference} established that the ``share'' posterior, $\rhohat^1_t$ (and thus the next round prior $\muhat_{t+1}$) must lie either on $\D$ or on the simplex boundary. Consider the latter scenario. A simplex boundary is defined by a set of states $\mathcal{B}$ such that for each $\thetahat \in \mathcal{B}$, $\rhohat(\thetahat) \in \{0,1\}$. Observe that once the prior is on the boundary, any signaling thereafter must also lie on the same boundary. Otherwise, there is an induced posterior for which a $\thetahat' \in \mathcal{B}$ satisfies $\rhohat_t(\thetahat'|s) \notin \{0,1\}$ and no convex combination with this posterior can lead to the prior which has $\muhat_{t+1}(\thetahat') \in \{0,1\}$. Observe also that the simplex boundary is itself a lower dimensional simplex. In other words, if optimal signaling moves the next round's prior to be on the boundary, the problem reduces to a problem over a smaller state-space of size at most $|\Theta| - 1 = n$, for which we can appeal to the inductive hypothesis. For the other scenario where the prior at round $t+1$ lies on the surface $\D$, we note that the induced posteriors of strictly optimal signaling at round $t+1$ must now be two points that lie on either $\D$ or at the boundary. As we have already addressed the boundary case, observe if the share posterior is on $\D$ and $dim(\D) = 0$, if the next round ($t+2$) can strictly improve utility, the posteriors must be on the boundary. If $dim(\D) \in [1,n-1]$, then both posteriors are either on (1) $\D$, or (2) on the boundary. For (1), observe that if one posterior lies in $\D$, the other must also since their linear combination must be $\muhat_{t+1} \in \D$. Since $\D$ has dimension $\leq n-1$ and all induced posteriors thereafter must stay on this lower dimensional plane, we can reformulate this as a lower dimensional problem and appeal to the induction hypothesis. Thus in at most 2 steps, we arrive at a lower dimensional problem. Thus, when the number of states is $|\Theta| = n+1$, we require at most $2 + 2(n-1) = 2n = 2(|\Theta| -1)$ iterations, confirming the induction hynothesis. 
\end{proof}

\subsection*{Proof of Theorem \ref{theorem:performative_monotonic}}
\begin{proof}
    First, note that scaling the platform utility by an additive constant does not affect optimal signaling, and thus we can solve the normalized instance. 
    We will show that the $\lambda = 0$ process under the stated conditions always leads to strictly increasing platform utility until reaching stability. Consider an arbitrary round $t \geq 0$, with the optimal user action being $a_{t} \in \{0,1\}$ for prior $\mu_t$. If we are not at a stable point, platform utility can always improve due to signaling. Thus:
    \begin{equation*}
         \E_{\mu_t}[\up'(a_t, \theta)] < P_t(s=0)\E_{\rho^0_t}[\up'(0, \theta)] + P_t(s=1)\E_{\rho^1_t}[\up'(1, \theta)]
    \end{equation*}
    where we use the fact that optimal signaling is persuasive. Next, inductively assume that the user's best action at belief $\mu_t$ is 1 (to share) which yields positive platform utility (we will see why this always holds). We know this is true for the first round by the theorem condition and the fact that user default action at $\mu_0$ is assumed to be share. Since platform utility for action $0$ (not share) is always less than or equal to 0 for any state under the theorem conditions: $\E_{\mu_t}[\up'(a_t, \theta)] < P_t(s=1)\E_{\rho_t^1}[\up'(1, \theta)]$, which implies
    \begin{equation*}\label{equation:performative_monotonic}
        \frac{1}{P_t(s=1)}\E_{\mu_t}[\up'(a_t, \theta)] < \E_{\rho_t^1}[\up'(1, \theta)] = \E_{\mu_{t+1}}[\up'(1, \theta)]
    \end{equation*} 
    Lastly, we note that the next round's prior is equal to the posterior induced by the share/1 signal. Since signaling is always persuasive, the optimal action of the receiver at belief $\rho_t^1 = \mu_{t+1}$ (wherein they receive the share signal) is to share, which has positive utility for the platform by the above. Thus the assumption we made earlier always holds for any round $t+1$. Note that since $P_t(s=1) \in [0,1]$, $\E_{\mu_t}[\up'(a_t, \theta)] < \E_{\mu_{t+1}}[\up'(a_{t+1}, \theta)] = \E_{\mu_{t+1}}[\up'(1, \theta)]$, and this process monotonically converges to a stable point. 
\end{proof}

\end{document}